\documentclass[runningheads]{llncs}

\usepackage{url} 
\usepackage{graphicx}

\usepackage{mymacros}
\usepackage{graphicx}
\usepackage{booktabs}
\usepackage{algorithm}
\usepackage{algpseudocode}
\usepackage{amsmath}
\usepackage{amsfonts}
\usepackage{amssymb}
\usepackage[numbers]{natbib}
\usepackage{tikz-cd}

\newcommand{\riccardo}[1]{{#1}}
\newcommand{\cancella}[1]{}

\begin{document}
\setcounter{page}{1}

\title{Exact Persistent Stochastic Non-Interference}



\author{Carla Piazza\inst{1}\orcidID{0000-1111-2222-3333} \and
Riccardo Romanello\inst{2}\orcidID{1111-2222-3333-4444} \and
Sabina Rossi\inst{2}\orcidID{2222--3333-4444-5555}}
\authorrunning{C. Piazza et al.}
%
\institute{
Universit\`a degli Studi di Udine, Italy \\
\email{carla.piazza@uniud.it} \and  
Universit\`a Ca' Foscari Venezia, Italy \\
\email{\{riccardo.romanello,sabina.rossi\}@unive.it}}

\maketitle


\begin{abstract}
Persistent Stochastic Non-Interference (PSNI) was introduced to capture a quantitative security property in stochastic process algebras, ensuring that a high-level process does not influence the observable behaviour of a low-level component, as formalised via lumpable bisimulation. In this work, we revisit PSNI from a performance-oriented perspective and propose a new characterisation based on a refined behavioural relation. We introduce \emph{weak-exact equivalence}, which extends exact equivalence with a relaxed treatment of internal (\(\tau\)) actions, enabling precise control over quantitative observables while accommodating unobservable transitions. Based on this, we define \emph{Exact PSNI} (EPSNI), a variant of PSNI characterised via weak-exact equivalence. We show that EPSNI admits the same bisimulation-based and unwinding-style characterisations as PSNI, and enjoys analogous compositionality properties. These results confirm weak-exact equivalence as a robust foundation for reasoning about non-interference in stochastic systems.


\end{abstract}

\section{Introduction}
\label{sec:intro}

Formal notions of non-interference have long been central to the specification and analysis of information flow properties in secure systems. 
Traditionally studied in deterministic or nondeterministic process models, non-interference ensures that high-level (secret) actions do not affect the observable behaviour of low-level (public) components. 
More recently, attention has turned to quantitative settings, where systems are described using stochastic process algebras, and the need arises to reason not only about qualitative properties but also about the performance and probabilistic characteristics of a system’s execution.
In this context, the recent work of Hillston, Piazza, and Rossi~\cite{PSNI_strong} introduced the notion of \emph{Persistent Stochastic Non-Interference} (PSNI), a security property formulated within the PEPA process algebra~\cite{hillston:book}. 
PSNI leverages the theory of lumpability for Markov chains to define a compositional and performance-aware form of non-interference. Specifically, PSNI ensures that the behaviour of a low-level process remains indistinguishable to an external observer regardless of the high-level environment in which it is placed. The notion is persistent in that it requires this indistinguishability to hold at every reachable state, not just at the initial configuration. PSNI admits multiple characterisations, including a bisimulation-based formulation and a set of unwinding-style local conditions, and it has been shown to be preserved under key operations of the PEPA calculus.
Building on this foundation, the aim of this work is to re-express PSNI using a more fine-grained behavioural equivalence that more directly captures the quantitative structure of stochastic systems. Our starting point is the classical notion of \emph{exact equivalence}, which requires two processes to exhibit the same apparent rates and to behave identically in all contexts. Exact equivalence provides a strong notion of behavioural indistinguishability, particularly suited to performance modeling, as it ensures that two processes can be interchanged without affecting any quantitative measure derived from the underlying Markov process. 
It thus seemed a natural candidate for reformulating the PSNI property in a more precise way.
However, applying exact equivalence directly to the setting of non-interference raises a subtle challenge. 
In particular, the presence of internal (\(\tau\)) actions—which are intended to be unobservable—clashes with the strict matching requirements of exact equivalence. When internal transitions differ across components, exact equivalence is too rigid to identify them, even when their observable behaviour is identical. 
As a consequence, exact equivalence fails to support the abstraction necessary for reasoning about security properties in systems where internal actions encode confidential or implementation-specific behaviour.
To overcome this limitation, we introduce a new equivalence relation, which we call \emph{weak-exact equivalence}. This relation extends exact equivalence by relaxing its requirements for \(\tau\)-transitions. 
While it continues to enforce matching of apparent rates and incoming transition rates for visible actions, it allows for a controlled discrepancy in internal behaviour, provided that such differences do not affect observability at the level of equivalence classes. 
The key intuition is that weak-exact equivalence preserves the precise quantitative structure of a system, while abstracting away from internal transitions in a way that remains compatible with security analysis.
With this new equivalence in place, we define a corresponding security property, \emph{Exact Persistent Stochastic Non-Interference} (EPSNI), which mirrors the structure of PSNI but uses weak-exact equivalence as its semantic foundation. 
We show that EPSNI admits two out of the three characterisations of PSNI: a semantic definition quantifying over high-level contexts and a bisimulation-based reformulation that avoids universal quantification. 
We further prove that EPSNI is preserved under PEPA operators such as prefix and hiding on low-level actions, and that it supports a notion of steady-state independence between high-level behaviour and low-level observations.
These results demonstrate that EPSNI provides a robust and expressive framework for reasoning about non-interference in stochastic systems, with strong guarantees on both performance preservation and security. Compared to PSNI, EPSNI offers a stricter form of behavioural indistinguishability, capturing finer-grained aspects of process interaction while maintaining compatibility with internal abstraction. At the same time, the similarity of its structural and compositional properties confirms that weak-exact equivalence is well-suited as a semantic foundation for quantitative security analysis in stochastic models.

\emph{Structure of the paper.} 
In Section \ref{sec:calculus}, we introduce the process algebra PEPA, its structural operational semantics, and the observation equivalence named \emph{weak-exact equivalence}. The notion of \emph{Exact Persistent Stochastic Non-Interference ($\mathit{EPSNI}$)} and its characterizations are presented in Section \ref{sec:psni}. In Section \ref{sec:compositionality}, we prove some compositionality results and other properties of $\mathit{EPSNI}$.
Section \ref{sec:conclusion} concludes the paper.

\section{The Calculus}\label{sec:calculus}
PEPA (Performance Evaluation Process Algebra)  \cite{hillston:book} is an
algebraic calculus  enhanced with stochastic timing infor\-mation
 which may~be used to~calculate performance measures 
as well as prove functional system~properties. 

The basic elements of PEPA are \emph{components} and \emph{activities}.
Each activity is represented by a pair $(\alpha, r)$ where $\alpha$ is a label, or \emph{action type}, 
and $r$ is its  \emph{activity rate},  that is the parameter of a negative
exponential distribution determining its duration. 
We assume that
there is a countable
 set, $\cA$,  of possible action types, including a distinguished
type, $\tau$,
which can be regarded as the \emph{unobservable} type.
Activity rates may be any positive real number, or the distinguished symbol
$\top$ which should be read as \emph{unspecified}.
The syntax for  PEPA terms is defined by the  grammar:

$$\begin{array}{cclccl}
N & ::= & (\alpha, r).N \mid N+N \mid X\\[1mm]
P & ::= & N \mid P/A \mid P \sync{A} P
\end{array}$$
where $N$ denotes a \emph{sequential component}, $X$ is a constant from
a countable set $\mathcal X$, $A\subset \mathcal A$ is a set of actions,
and $P$ denotes a
\emph{model component}. 
We write $\cC$ for the set of all possible
components.\\[-3mm]

\begin{table*}[t]
	\begin{center}
	\begin{tabular}{ccc}
		\toprule
		& \\
 \multicolumn{3}{c}
{$\dfrac{}{(\alpha,r).N \transits{(\alpha,r)} N}$ \qquad $\dfrac{N \transits{(\alpha,r)} N'}{X \transits{(\alpha,r)} N'}$  $(X\rmdef N)$}\\[8mm]
  \multicolumn{3}{c}
 {
$\dfrac{N \transits{(\alpha,r)} N'}{N+M \transits{(\alpha,r)} N'}$ \qquad
  $\dfrac{M \transits{(\alpha,r)} M'}{N+M \transits{(\alpha,r)} M'}$}
  \\[8mm]
  \multicolumn{3}{c}
 { $\dfrac{P \transits{(\alpha,r)} P'}{P/A \transits{(\alpha,r)} P'/A}$  $(\alpha  \not \in A)$ \qquad
$\dfrac{P \transits{(\alpha,r)} P'}{P/A \transits{(\tau,r)} P'/A}$  $(\alpha  \in A)$ }
\\[8mm]
\multicolumn{3}{c}
{ $\dfrac{P \transits{(\alpha,r)} P'}{P \sync{A} Q \transits{(\alpha,r)} P'\sync{A} Q}$  $(\alpha \not \in A)$ \qquad
 $\dfrac{Q \transits{(\alpha,r)} Q'}{P \sync{A} Q \transits{(\alpha,r)} P\sync{A} Q'}$ $(\alpha \not \in A)$}\\[8mm]
 \multicolumn{3}{c}
{ $\dfrac{P \transits{(\alpha,r_1)} P' \ \ Q \transits{(\alpha,r_2)} Q'}{P\sync{A} Q \transits{(\alpha,R)} P'\sync{A} Q'}$ \ \ \   $R =  \dfrac{r_1}{r_{\alpha}(P)} \dfrac{r_2}{r_{\alpha}(Q)} \ \mathrm{min}(r_{\alpha}(P),r_{\alpha}(Q))$ \ $(\alpha  \in A)$}

\\
	\bottomrule 
\end{tabular}
\end{center}
\caption{Operational semantics for PEPA components}
\label{table:rules}
\end{table*}

\noindent
\emph{Structural Operational Semantics.} 
Table \ref{table:rules} gives the structural operational semantics for PEPA.

Component $(\alpha, r).N$ carries out the activity $(\alpha, r)$
of type  $\alpha$ at rate $r$ and subsequently behaves as  $N$. When $a=(\alpha, r)$, component $(\alpha, r).N$  may be written as $a.N$.

Component $N+M$ represents a system which may behave either as  $N$ or as  $M$. $N+M$ enables all the current activities of both $N$ and $M$. The first activity to complete decides whether $N$ or $M$ will proceed. The other component of the choice is discarded. 

The meaning of a constant $X$ is given by a defining equation such as 
$X\rmdef N$ which gives the constant $X$ the behaviour of the component~$N$. 
For each constant exactly one defining equation has to be specified.
We require that 
constant definitions are always guarded (e.g., a definition of the form $X\rmdef Y + (\alpha,r).Z$ is not allowed since $Y$ is not guarded). 
Component $P/A$ behaves as $P$ except that any activity of type within the set $A$ is \emph{hidden}, i.e., they are relabeled with the unobservable type $\tau$. 
The cooperation combinator $\sync{A}$ is in fact an indexed family of combinators, one
for each possible set of action types, $A\subseteq \cA\setminus \{\tau\}$.
The \emph{cooperation set} $A$  defines the action types on which the components must synchronize or \emph{cooperate} (the unobservable action type, $\tau$, may not appear in any cooperation set).
It is assumed that each component proceeds independently with any activity whose type does not occur
in the cooperation set $A$ (\emph{individual activities}). However, activities with action types in the set $A$ require
the simultaneous involvement of both components (\emph{shared activities}). These shared activities will only
be enabled in $P\sync{A}Q$ when they are enabled in both $P$ and $Q$. 

The shared activity will have the same action type as the two contributing activities and a rate
reflecting the rate of the slower participant \cite{hillston:book}.

If an activity has an unspecified rate in a sequential component, then the
component is passive~with respect to that action type and all occurrences of the activity in 
the sequential component have to be unspecified. In this case, 
the rate of 
the shared activity will be completely determined by the
other component.
For a given  process $P$ and action  type $\alpha$, 
the \emph{apparent rate}  of 
$\alpha$ in $P$,
denoted $r_{\alpha}(P)$  is the sum of the rates of the 
$\alpha$  activities enabled in~$P$.

Both in the last rule of Table \ref{table:rules} and in the definition of apparent rate, 
when operations involving $\top$ are concerned, we refer to the rules defined
in \cite{hillston:book} Section 3.3.5, stating for instance that $r<\top$ for any positive real number 
$r$.

The semantics of each
term in PEPA is given via a labeled \emph{multi-transition system} where the 
multiplicities of arcs are significant. In the transition system, a
state or
\emph{derivative} corresponds to each syntactic term of the language 
and an arc represents the activity which causes
one derivative to evolve into another. The  set of
reachable states of a model $P$ is termed the \emph{derivative set} of $P$, denoted by $ds(P)$,
 and
constitutes the set of nodes of the \emph{derivation graph} of $P$ ($\cD(P)$) obtained
by applying the semantic rules exhaustively.

Since recursive definitions are limited to sequential guarded components, each PEPA component $P$ has a derivation graph $\cD(P)$ with a finite number of states and edges. In \cite{CCSfinite1} the property is proved for a fragment of CCS and the same arguments apply here.

We denote by $\cA(P)$ the set of all the \emph{current action types} of $P$, i.e.,
 the set of action types which the component $P$ may next engage in.
We denote by $\Ac(P)$ the multiset of all the \emph{current activities} of $P$. 

Finally, we denote by ${\vec{\cal A}(P)}$ the union of all ${\cA{(P')}}$ with $P'\in ds(P)$, i.e., the set of all action types syntactically occurring in $P$ and in the constants involved in the definition of $P$.
For any component $P$, the \emph{exit rate} from $P$ will be the sum of the 
activity rates of all the activities
enabled in $P$, i.e., $ q(P) = \sum_{a \in \Ac(P)}r_a$, with
$r_a$ being the rate of activity $a$. 
If $P$ enables more than one activity, $|\Ac(P)|>1$, then the
dynamic behaviour of the model is determined by a race condition. 
This has the effect of replacing the
nondeterministic branching of the pure process algebra with probabilistic 
branching. The probability
that a particular activity completes is given by the ratio of the activity 
rate to the exit rate from $P$.\\[-3mm]

\noindent
\emph{Underlying Stochastic Process.} 
We say that a PEPA model $P$ is \emph{complete} if all the exit rates of the derivatives of $P$ are real positive numbers.
Intuitively, if $P$ is complete and there is a passive action occurring in one of its sub-components, then
it is under the scope of a cooperation which either restricts the action or forces the synchronization with an active occurrence of the same action.
In \cite{hillston:book} Section 3.5 it is proved that
for any finite complete PEPA model $P$ with $ds(P)=\{P_0,\ldots,P_n\}$, if we define the stochastic process 
$X(t)$, such that $X(t)=P_i$ indicates that the system behaves as component 
$P_i$ at time t, then $X(t)$ is a continuous time
Markov chain.
The \emph{transition rate} between two components $P_i$ and $P_j$, denoted 
 $q(P_i,P_j)$,~is the rate at
which the system changes from behaving as component $P_i$ to behaving as $P_j$. 
It is the sum of the
activity rates labeling arcs which connect the node cor\-responding to $P_i$ 
to the node corresponding to $P_j$ in $\cD(P)$,~i.e., 
$ q(P_i,P_j)= \sum_{a\in \Ac(P_i| P_j)} r_a$
where $P_i\not = P_j$ and $\Ac(P_i| P_j) = \bms a\in \Ac(P_i) |\ P_i \transits{a} P_j\ems$. 
Clearly, if $P_j$ is not a one-step derivative of $P_i$, $q(P_i,P_j)=0$.
The $q(P_i,P_j)$ (also denoted $q_{ij}$), are the off-diagonal elements of the 
infinitesimal generator matrix of the
Markov process, ${\bf Q}$. Diagonal elements are formed as the negative sum of 
the non-diagonal elements of
each row. We use the following notation: $q(P_i)=\sum_{j\neq i}q(P_i,P_j)$ and
 $q_{ii}=-q(P_i)$. 
For any finite and irreducible PEPA model $P$, the steady-state distribution 
$\mathrm{\Pi}(\cdot)$ exists and it may be found by solving the normalization equation and the global
balance equations:
\[
\sum_{P_i\in ds(P)}\mathrm{\Pi}(P_i)=1 
\quad \land \quad
  \mathrm{\Pi} \, \mathbf{Q} =\mathbf{0}.
\]

The \emph{conditional transition rate} from $P_i$ to $P_j$ via an action type 
$\alpha$ is denoted $q(P_i,P_j,\alpha)$. This is
the sum of the activity rates labeling arcs connecting the corresponding 
nodes in the derivation graph with label
 $\alpha$, i.e., $q(P_i,P_j,\alpha)=\sum_{P_i \xrightarrow{(\alpha,r_{\alpha})} P_j} r_{\alpha}$.
 It is the rate at which a system behaving as component $P_i$
evolves to behaving as component $P_j$ as the result of completing a 
type $\alpha$ activity. 
The notions of both transition rate and conditional transition rate can 
be extended to any process $Q\not \in ds(P)$ by defining $q(P_i,Q)=q(P_i,Q,\alpha)=0$.

The \emph{total conditional transition rate} from $P$ to $S\subseteq \mathcal C$,
denoted $q[P,S,\alpha]$, is defined as
$
q[P,S,\alpha]=\sum_{P'\in S} q(P,P',\alpha)$.
\riccardo{Symmetrically, we define the transition from $S \subseteq \mathcal{C}$ to $P$ as $q[S, P, \alpha] = \sum_{P'\in S} q(P',P,\alpha)$}. 
Moreover, we define $q[P, a] = \sum_{P' \in \mathcal{C}} q(P, P', \alpha)$ as the sum of the total outgoing rates from state $P$ to any component of the system.

For further explanation and examples on PEPA we refer the reader to \cite{hillston:book}.

\noindent
\emph{Observation Equivalence.} 

\riccardo{
n~\cite{FG95}, the notion of Stochastic Non-Interference was introduced as a security property for process algebras. This property ensures that secret (high-level) actions do not influence the observable behaviour of a system in a probabilistic setting. Later, in~\cite{PSNI_strong}, the concept of Persistency was incorporated, requiring that the non-interference property hold not only initially but at all reachable states of the system. The resulting notion was termed Persistent Stochastic Non-Interference (PSNI).

The key idea behind these properties is the inability of an observer to detect any information flow through protected (i.e., confidential) actions. To formalise this behaviour, the authors of~\cite{PSNI_strong} adopted an equivalence relation known as lumpable bisimilarity~\cite{marin:valuetools13}, defined over PEPA components. This relation was chosen because it induces a strong equivalence over the underlying Continuous-Time Markov Chain (CTMC), preserving both behavioural and performance properties.

In this work, we revisit the approach proposed in~\cite{PSNI_strong}, shifting the focus from strong equivalence to exact equivalence~\cite{franceschinis:bounds}. While both notions are defined over CTMCs, their perspectives differ: strong equivalence compares the rates leaving a state, whereas exact equivalence focuses on the rates entering a state. Our goal is to reformulate PSNI using this more precise, incoming-rate–oriented view of process behaviour.

Since exact equivalence is originally defined at the CTMC level, a key step in our development is to lift this notion to the level of PEPA components. 

Two PEPA components are said to be \emph{exact equivalent} if there is an equivalence relation between them such that,
for any action type $\alpha$ the total conditional transition rates from any equivalence class to those components, via activities of this type, are the same.
Moreover, it must hold that the total outgoing rates are the same. 

\begin{definition}\emph{(Exact equivalence)}
    \label{def::exact_equiv}
    An equivalence relation over PEPA components, $\cR\subseteq \cC\times \cC$ is an \emph{exact equivalence} if whenever $(P, Q) \in \cR$ then for all $\alpha\in \cA$ and for all  $S\in \cC/\cR$ it holds: 
    \begin{itemize}
       \item $q[P, \alpha] = q[Q, \alpha] $
        \item $q[S, P, a] = q[S, Q, a]$
    \end{itemize}
    
\end{definition}
}

\riccardo{It is clear that the identity relation is an exact equivalence. We are interested in the relation which is the largest exact equivalence, formed by the union of all exact equivalences. }

\riccardo{
It was proven in \cite{inf18} that the transitive closure of a union of exact equivalences is still an exact equivalence. Hence, the maximal exact equivalence, denoted $\sim_{e}$, is defined as the
union of all exact equivalences and it is an equivalence relation.}
\riccardo{
\section{Weak-Exact Equivalence}

The equivalence relation we just introduced does not make any difference between $\tau$ and non-$\tau$ actions. 
However, in \cite{PSNI_strong}, the very first step was to tweak Strong equivalence in order to behave differently with the unobservable actions. 

For this reason, lumpable bisimulation was introduced. 

\begin{definition}\emph{(Lumpable bisimulation)}
  \label{def:bisimulation}
  An equivalence relation over PEPA components, $\cR\subseteq \cC\times \cC$, is a \emph{lumpable bisimulation} if whenever $(P,Q)\in \cR$ then for  all $\alpha\in \cA$ and for all  $S\in \cC/\cR$ such that
  \begin{itemize}
  \item either $\alpha\not = \tau$,
  \item   or $\alpha=\tau$ and $P,Q\not \in S$,
  \end{itemize}
  it holds
  $$q[P,S,\alpha]=q[Q,S,\alpha]\, .$$
\end{definition}

Since in this work we want to approach the very same problem but starting from the notion of Exact Equivalence we stated above in Def. \ref{def::exact_equiv}, 
it is straightforward to see that the equivalence must be tweaked to take into account the behaviour of $\tau$-labelled actions, mimicking what 
Def. \ref{def:bisimulation} does with Strong Equivalence \cite{hillston:book}.

Therefore, we introduce the following notion which we call \emph{weak} exact equivalence.

\begin{definition}[weak-exact equivalence]
    An equivalence relation over PEPA components, $\cR\subseteq \cC\times \cC$ is a weak-exact equivalence if whenever $(P, Q) \in \cR$ then for all $\alpha \in \mathcal{A}$:
    \begin{itemize}
        \item either $\alpha \neq \tau$ and: 
        \begin{itemize}
            \item $q[P, \alpha] = q[Q, \alpha]$
            \item $\forall S \in \mathcal{C}/\mathcal{R} \,\,\,\,\,\, q[S, P, \alpha] = q[S, Q, \alpha]$
        \end{itemize}
        \item or $\alpha = \tau$ and:
        \begin{itemize}
            \item $\forall S \in \mathcal{C} / \mathcal{R}, P, Q \notin S \,\,\,\,\,\, q[S, P, \tau] = q[S, Q, \tau]$
            \item $P, Q \in S \,\,\,\,\,\, q[S, P, \tau] - q[P, \tau] = q[S, Q, \tau] - q[Q, \tau]$
        \end{itemize}
    \end{itemize}
\end{definition}

While this new relation may seem counterintuitive, its properties have been derived straightforwardly from the sufficient conditions 
for the exact equivalence stated in \cite{franceschinis:lumping.peva}.

The new relation takes into account both the conditions involved in exact equivalence. However, 
it differently deals with $\tau$-transitions as well.

\begin{definition}\emph{(Weak-Exact equivalence)}
\label{def:bisimilarity}
Two PEPA components $P$ and $Q$ are \emph{weak-exact equivalent}, 
written 
$P\approx_{\text{we}}Q$, if $(P,Q)\in\cR$ for some weak-exact equivalence $\cR$,~i.e.,
$$\approx_{\text{we}} \ =\bigcup \ \{\cR\ |\ \cR \mbox{ is a weak-exact equivalence}\}.$$
\noindent
$\approx_{\text{we}}$ is called \emph{weak-exact equivalence} and it is the largest symmetric weak-exact equivalence over PEPA components.
\end{definition}

To the best of our knowledge, this is the very first time such notion is adopted in the literature. 
Therefore, it is fundamental to prove its compositional properties. 

To obtain such endeavour, we first define the notion of equivalence \emph{up to}. 
Our goal is to mimic some techniques adopted in past works like \cite{hillston:thesis}.

\begin{definition}
  \label{def::up-to}
  $\mathcal{R}$ is a weak-exact equivalence up to $\approx_{\text{we}}$ if $\mathcal{R}$ is an equivalence relation over $\mathcal{C}$ and $(P, Q) \in \mathcal{R}$
  implies that for all $\alpha \in \mathcal{A}$:
  \begin{itemize}
    \item either $\alpha \neq \tau$ and: 
        \begin{itemize}
            \item $q[P, \alpha] = q[Q, \alpha]$
            \item $\forall T \in \mathcal{C}/(\approx_{\text{we}}\mathcal{R}\approx_{\text{we}}) \,\,\,\,\,\, q[T, P, \alpha] = q[T, Q, \alpha]$
        \end{itemize}
    \item or $\alpha = \tau$ and:
        \begin{itemize}
            \item $\forall T \in \mathcal{C} / (\approx_{\text{we}}\mathcal{R}\approx_{\text{we}}), P, Q \notin T \,\,\,\,\,\, q[T, P, \tau] = q[T, Q, \tau]$
            \item $P, Q \in T \,\,\,\,\,\, q[T, P, \tau] - q[P, \tau] = q[T, Q, \tau] - q[Q, \tau]$
        \end{itemize}
  \end{itemize}
\end{definition}

We recall that $P \approx_{\text{we}}\mathcal{R}\approx_{\text{we}} Q$ if and only if there exist $P_1, Q_1$ such that $P \approx_{\text{we}} P1$, 
$P1 \mathcal{R} Q_1$, and $Q_1 \approx_{\text{we}} Q$.
In some sense, we are considering different equivalence classes induced by $\approx_{\text{we}}$ and joining them via $\mathcal{R}$. 
Because of this very reason, the classes induced by $\approx_{\text{we}}\mathcal{R}\approx_{\text{we}}$ are actually union of classes induced by $\approx_{\text{we}}$.

The above definition turns out to be really useful in proving properties of weak-exact equivalence by looking at smaller relations. 
Such behaviour is proved through the following lemma.

\begin{lemma}
  \label{lemma::up-to}
  If $\mathcal{R}$ is a weak-exact equivalence up to $\approx_{\text{we}}$, then the relation $\approx_{\text{we}} \mathcal{R} \approx_{\text{we}}$ is a weak-exact equivalence. 
\end{lemma}

The up to property turns out to be fundamental when dealing with compositionality of weak-exact equivalence. 
\begin{proposition}
  \label{prop::compositionality}
  If $P_1 \approx_{\text{we}} P_2$ then $P_1 \setminus L \approx_{\text{we}} P_2 \setminus L$.
\end{proposition}
}

\section{Exact Persistent Stochastic Non-Interference}\label{sec:psni}

The security property named \emph{Persistent Stochastic Non-Interference ($\mathit{PSNI}$)} was introduced in~\cite{PSNI_strong}. 
Informally, it aims to capture every possible information flow from a \emph{classified (high)} level of confidentiality to an \emph{untrusted (low)} one.
The property was derived from the classical notion of Non-Interference, enriched with the idea of \emph{persistency}, meaning that the security property must hold for all reachable processes of the system.

In their proposal, the authors defined PSNI as follows. First, they partitioned the set $\cA \setminus \{\tau\}$ of visible action types into two disjoint sets, ${\cal H}$ and ${\cal L}$, representing high- and low-level action types, respectively. A high-level PEPA component $H$ is a PEPA term such that for all $H' \in ds(H)$, $\cA(H') \subseteq {\cal H}$, i.e., every derivative of $H$ may engage only in high-level actions. The set of all such high-level PEPA components is denoted by ${\cal C}_H$. Additionally, they extended the PEPA grammar by adding the {\boldmath$0$} symbol to denote a component that performs no activity.

With these notions in place, the authors introduced the following definition of Stochastic Non-Interference:

\begin{definition}\emph{(Stochastic Non-Interference, SNI)}
  Let $P$ be a PEPA component.
  \[
  P \in \mathit{SNI} \mbox{ iff }
  \forall H \in {\mathcal C}_H,\
  (P \sync{{\mathcal H}} {\bf 0})/{\mathcal H} \approx_l (P \sync{{\mathcal H}} H)/{\mathcal H}\,.\\[2mm]
  \]
\end{definition}

Here, $\approx_l$ denotes the equivalence relation known as \emph{lumpable bisimilarity}, which corresponds to the union of all possible lumpable bisimulations (see Definition~\ref{def:bisimulation}).

However, this definition did not enforce propagation of the security property across all process derivatives. To address this, the authors introduced the notion of \emph{reachable} processes and extended the property to account for persistency, yielding the definition of Persistent SNI:

\begin{definition}\emph{(Persistent Stochastic Non-Interference)}
Let $P$ be a PEPA component. 
\[
P \in \mathit{PSNI} \mbox{ iff }
\forall P' \in ds(P), \; P' \in \mathit{SNI}
\]
\end{definition}

In~\cite{PSNI_strong}, the following characterisation of PSNI was also provided:
\[
P \in \mathit{PSNI} \mbox{ iff }
\forall P' \in ds(P), \; \forall H \in {\mathcal C}_H,\
P' \setminus {\mathcal H} \approx_{\text{we}} (P' \sync{{\mathcal H}} H)/{\mathcal H}\,.\\[2mm]
\]

It is evident that all the above definitions fundamentally rely on the equivalence relation called \emph{lumpable bisimilarity}, which is built upon the notion of lumpable bisimulation.

However, as discussed in previous sections, alternative equivalence relations can also be considered. The aim of this paper is to characterise a new security property, which we call \emph{Exact PSNI}, abbreviated as $\mathit{EPSNI}$.

We begin by adapting the definition of SNI to obtain a new property, \emph{Exact SNI}:

\begin{definition}\emph{(Exact Stochastic Non-Interference, ESNI)}
  Let $P$ be a PEPA component.
  \[
  P \in \mathit{ESNI} \mbox{ iff }
  \forall H \in {\mathcal C}_H,\
  (P \sync{{\mathcal H}} {\bf 0})/{\mathcal H} \approx_{\text{we}} (P \sync{{\mathcal H}} H)/{\mathcal H}\,.\\[2mm]
  \]
\end{definition}

Following the same reasoning applied in the definition of $\mathit{PSNI}$, we introduce the persistency aspect, resulting in the definition of \emph{Exact Persistent Stochastic Non-Interference} ($\mathit{EPSNI}$):

\begin{definition}\emph{(Exact Persistent Stochastic Non-Interference)}
\label{def:psni}
Let $P$ be a PEPA component. 
\[
P \in \mathit{EPSNI} \mbox{ iff }
\forall P' \in ds(P), \; P' \in \mathit{ESNI}
\]
\end{definition}

Notably, weak-exact equivalence induces an exact equivalence on the underlying Markov chain of the system. This follows directly from the fact that all conditions defining weak-exact equivalence over PEPA components correspond to sufficient conditions for exact equivalence on Markov chains, as presented in~\cite{franceschinis:bounds}. Therefore, a weak-exact equivalence on PEPA components guarantees exact equivalence in the associated CTMC.

This property implies that all states of the Markov chain belonging to the same equivalence class under $\approx_{\text{we}}$ are equiprobable in the stationary distribution of the chain.

The use of the hiding operator on both sides of the equivalence ensures that indistinguishability is defined with respect to the low-level view of the system. As a result, and due to the properties of $\approx_{\text{we}}$, we conclude that for a low-level observer, the probability of being in a particular state does not depend on high-level interactions.

Indeed, high-level activities may be present in the process $(P \sync{{\mathcal H}} H)$, but are absent in $(P \sync{{\mathcal H}} {\bf 0})$. This does not mean the system lacks high-level actions altogether, but rather that such actions cannot be inferred by a low-level observer.

Observe that $(P' \sync{{\mathcal H}} {\bf 0})$ cannot perform any high-level activity. Therefore, the derivation graphs of $(P' \sync{{\mathcal H}} {\bf 0})$ and $(P' \sync{{\mathcal H}} {\bf 0})/{\mathcal H}$ are isomorphic as labelled transition systems. More precisely, $(P' \sync{{\mathcal H}} {\bf 0})$ behaves as the component $P'$ constrained from engaging in any high-level action.

Thus, we adopt the notation $P \setminus {\mathcal H}$ to denote the system $(P' \sync{{\mathcal H}} {\bf 0})/{\mathcal H}$, in analogy with the restriction operator in CCS.

From this, we can restate the property as:
\[
P \in \mathit{EPSNI} \mbox{ iff }
\forall P' \in ds(P), \; \forall H \in {\mathcal C}_H,\
P' \setminus {\mathcal H} \approx_{\text{we}} (P' \sync{{\mathcal H}} H)/{\mathcal H}\,.\\[2mm]
\]

Our first result is that $\mathit{EPSNI}$ is preserved by \riccardo{weak-exact equivalence}.

\riccardo{
\begin{lemma}\label{first}
If $P$ and $Q$ are two PEPA components such that $P\approx_{\text{we}} Q$ and $P\in \mathit{EPSNI}$ then also $Q\in \mathit{EPSNI}$.
\end{lemma}
}

Subsequently, we introduce a novel bisimulation-based equivalence relation over PEPA components, named \riccardo{$\approx_{\text{we}}^{hc}$}, that allows us to give a first characterization of $\mathit{EPSNI}$ in which the quantification over all possible high level components $H$ is hidden inside the 
equivalence. In particular, we show that $P\in \mathit{EPSNI}$ if and only if $P\setminus {\cal H}$ and $P$ are not distinguishable with respect to $\approx_{\text{we}}^{hc}$. Intuitively, two processes are 
$\approx_{\text{we}}^{hc}$-equivalent if they can simulate each other in any possible high context, i.e., in every context $C[\_]$ of the form $(\_ \ \sync{{\cal H}}H)/{\cal H}$ where $H\in {\cal C}_H$.
Observe that for any high context $C[\_]$ and PEPA model $P$, all the states reachable from $C[P]$ have the form $C'[P']$ with $C'[\_]$ being a high context too and $P'\in ds(P)$.

We now introduce the concept of \emph{\cancella{lumpable bisimulation on high contexts}\riccardo{weak-exact equivalence on high contexts}}: the idea is that, given two PEPA models $P$ and $Q$, when a high level context $C[\_]$ filled with $P$ executes a certain activity moving $P$ to $P'$ then the same context filled with $Q$ is able to simulate this step moving $Q$ to $Q'$ so that $P'$ and $Q'$ are again \cancella{lumpable bisimilar} \riccardo{weak-exact equivalent} on high contexts, and vice-versa. This must be true for every possible high context $C[\_]$. It is important to note that the quantification over all possible high contexts is re-iterated for $P'$ and $Q'$.
%
For a PEPA model $P$,  $\alpha\in \cA$,  $S\subseteq \mathcal C$ and a high context $C[\_]$
we define
%
\(
 q_C(P,P',\alpha)=\sum_{C[P] \xrightarrow{(\alpha,r_{\alpha})} C'[P']} r_{\alpha}
\)
and
\cancella{\(
q_C[P,S,\alpha]=\sum_{P'\in S} q_C(P,P',\alpha)\,.
\)}
\riccardo{\(
q_C[S,P,\alpha]=\sum_{P'\in S} q_C(P',P,\alpha)\,.
\)}

The notion of \emph{\cancella{lumpable bisimilarity}\riccardo{weak-exact equivalence} on high contexts} is defined as follows.
\riccardo{
\begin{definition}\emph{(weak-exact equivalence on high contexts)}
\label{def:hcontexts}
An equivalence relation over PEPA components, $\cR\subseteq \cC\times \cC$, is a \emph{weak-exact equivalence on high contexts} if whenever $(P,Q)\in \cR$ then for all high context $C[\_]$,
for  all $\alpha\in \cA$
\begin{itemize}
  \item either $\alpha \neq \tau$ and: 
  \begin{itemize}
    \item $q_C[P, \alpha] = q_C[Q, \alpha]$
    \item $\forall S \in \mathcal{C}/\mathcal{R} \,\,\,\,\,\, q_C[S, P, \alpha] = q_C[S, Q, \alpha]$
  \end{itemize}
  \item or $\alpha = \tau$ and:
  \begin{itemize}
    \item $\forall S \in \mathcal{C} / \mathcal{R}, P, Q \notin S \,\,\,\,\,\, q_C[S, P, \tau] = q_C[S, Q, \tau]$
    \item $P, Q \in S \,\,\,\,\,\, q_C[S, P, \tau] - q_C[P, \tau] = q_C[S, Q, \tau] - q_C[Q, \tau]$
  \end{itemize}
\end{itemize}

Two PEPA components $P$ and $Q$ are \emph{weak-exact equivalent on high contexts}, 
written 
$P\approx^{hc}_{\text{we}}Q$, if $(P,Q)\in\cR$ for some weak-exact equivalence on high contexts~$\cR$,~i.e.,
$$\approx^{hc}_{\text{we}} \ =\bigcup \ \{\cR\ |\ \cR \mbox{ is a weak-exact equivalence on high contexts}\}.$$
\noindent
$\approx^{hc}_{\text{we}}$ is called \emph{weak-exact equivalence on high contexts} and it is the largest weak-exact equivalence on high contexts over PEPA components. 
\end{definition}
}

Theorem~\ref{theo:hc} gives a characterization of $\mathit{EPSNI}$ in terms of  \cancella{$\approx^{hc}_l$}\riccardo{$\approx^{hc}_{\text{we}}$}. Notice that
\cancella{$\approx^{hc}_l$}\riccardo{$\approx^{hc}_{\text{we}}$} still contains a universal quantification over all high level contexts. However, the characterization of $\mathit{EPSNI}$ given in Theorem \ref{theo:hc} gives insights into the meaning and the properties of $\mathit{EPSNI}$. As a matter of fact, the persistency requirement is now absorbed in the notion of \cancella{lumpable bisimulation} \riccardo{weak-exact equivalence} on high contexts which is in a sense persistent by definition, as usual for lumpability notions.

\riccardo{
\begin{theorem}\label{theo:hc}
Let $P$ be a PEPA component. Then \[P\in \mathit{EPSNI}\mbox{ iff } P\setminus{\cal H}\approx^{hc}_{\text{we}} P\,.\]
\end{theorem}

}

We show how it is possible to give a characterization of $\mathit{EPSNI}$ avoiding both the universal quantification over all the possible high level components and the universal quantification over all the possible reachable states.

Previously, we have shown how the idea of ``being secure in every state'' can be directly moved  inside the \cancella{lumpable bisimulation}\riccardo{weak-exact equivalence} on high contexts notion ($\approx^{hc}_{\text{we}}$). However, this bisimulation notion implicitly contains a quantification over all possible high contexts. We now prove that $\approx^{hc}_{\text{we}}$ can be expressed in a rather simpler way by exploiting local information only. This can be done by defining a novel equivalence relation which focuses only on observable actions that do not belong to ${\cal H}$. More in detail, we define an observation equivalence where actions from ${\cal H}$ \emph{may} be ignored.

We first introduce the notion of \emph{\cancella{lumpable bisimilarity} \riccardo{weak-exact equivalence} up to ${\cal H}$}.
\riccardo{

\begin{definition}\emph{(weak-exact equivalence up to ${\cal H}$)}
\label{def:upto}
An equivalence relation over PEPA components, $\cR\subseteq \cC\times \cC$, is a \emph{weak-exact equivalence up to ${\cal H}$} if whenever $(P,Q)\in \cR$ then for all $\alpha\in \cA$:
\begin{itemize}
  \item either $\alpha \notin \mathcal{H} \cup \{\tau\}$ and: 
  \begin{itemize}
      \item $q[P, \alpha] = q[Q, \alpha]$
      \item $\forall S \in \mathcal{C}/\mathcal{R} \,\,\,\,\,\, q[S, P, \alpha] = q[S, Q, \alpha]$
  \end{itemize}
  \item or $\alpha \in \mathcal{H} \cup \{\tau\}$ and:
  \begin{itemize}
      \item $\forall S \in \mathcal{C} / \mathcal{R}, P, Q \notin S \,\,\,\,\,\, q[S, P, \alpha] = q[S, Q, \alpha]$
      \item $P, Q \in S \,\,\,\,\,\, q[S, P, \alpha] - q[P, \alpha] = q[S, Q, \alpha] - q[Q, \alpha]$
  \end{itemize}
\end{itemize}

Two PEPA components $P$ and $Q$ are \emph{weak-exact equivalent up to ${\cal H}$}, 
written 
$P\approx^{\cal H}_{\text{we}}Q$, if $(P,Q)\in\cR$ for some exact equivalence up to ${\cal H}$,~i.e.,
$$\approx^{\cal H}_{\text{we}} \ =\bigcup \ \{\cR\ |\ \cR \mbox{ is a weak-exact equivalence up to \mbox{${\cal H}$}}\}.$$
\noindent
$\approx^{\cal H}_{\text{we}}$ is called \emph{weak-exact equivalence  up to ${\cal H}$} and it is the largest symmetric weak exact equivalence up to ${\cal H}$ over PEPA components.
\end{definition}}

Theorem~\ref{teo:H-hc} shows that binary relations $\approx^{hc}_{\text{we}}$ and 
$\approx^{{\cal H}}_{\text{we}}$ are equivalent.
\riccardo
{
\begin{theorem}\label{teo:H-hc}
Let $P$ and $Q$ be two PEPA components. Then \(P \approx^{hc}_{\text{we}}  Q\) if and only if \(
P \approx^{\cal H}_{\text{we}} Q.\)
\end{theorem}
}

Theorem \ref{teo:H-hc} allows us to identify a local property of processes (with no quantification on the states and on the high contexts) which is a necessary and sufficient condition for  $\mathit{EPSNI}$. This is stated by the following corollary:

\riccardo{
\begin{corollary}
  \label{cor:psni}
Let $P$ be a PEPA component. Then \(P\in \mathit{EPSNI} \mbox{ iff } P\setminus {\cal H}\approx^{\cal H}_{\text{we}} P\,.\)
\end{corollary}

We now study some compositionality properties of weak-exact equivalence up to $\mathcal H$. These will allow us both to provide characterizations of $\mathit{EPSNI}$ in terms of unwinding conditions and to prove compositionality results for $\mathit{EPSNI}$.

\begin{lemma}\label{lemma:alto}
Let $\mathcal T$ be a symmetric and transitive binary relation over $\mathcal C$ and let $\mathcal T^*_{\mathcal H}$ be the symmetric and transitive closure
of $\mathcal T\cup \approx_{\text{we}}^{\mathcal H}$. $\mathcal T^*_{\mathcal H}$ is a weak-exact equivalence up to $\mathcal H$ if and only if for all 
$(P,Q)\in \mathcal T$ it holds that:
\begin{itemize}
  \item either $\alpha \notin \mathcal{H} \cup \{\tau\}$ and: 
  \begin{itemize}
      \item $q[P, \alpha] = q[Q, \alpha]$
      \item $\forall S \in \mathcal{C}/\mathcal T^*_{\mathcal H} \,\,\,\,\,\, q[S, P, \alpha] = q[S, Q, \alpha]$
  \end{itemize}
  \item or $\alpha \in \mathcal{H} \cup \{\tau\}$ and:
  \begin{itemize}
      \item $\forall S \in \mathcal{C} / \mathcal T^*_{\mathcal H}, P, Q \notin S \,\,\,\,\,\, q[S, P, \alpha] = q[S, Q, \alpha]$
      \item $P, Q \in S \,\,\,\,\,\, q[S, P, \alpha] - q[P, \alpha] = q[S, Q, \alpha] - q[Q, \alpha]$
  \end{itemize}
\end{itemize}
\end{lemma}

}

\riccardo{

We recall that $\mathcal L=\mathcal A\setminus \{\tau\}\cup \mathcal H$ is the set of low level action types.
\begin{lemma}\label{new1}
Let $P_1,P_2,Q,$ be PEPA components. Let $\alpha\in \mathcal A$, $A\subseteq \mathcal A$, and $L\subseteq \mathcal L$.
If $P_1\approx_{\text{we}}^{\mathcal H}P_2$, then: 
  \begin{itemize}
  \item[1.] $P_1\setminus A\approx_{\text{we}}^{\mathcal H}P_2\setminus A$;
  \item[2.] $P_1/ A \approx_{\text{we}}^{\mathcal H} P_2 / A$;
  \end{itemize}
\end{lemma}

}

\riccardo{
The following relationships between $\approx_{\text{we}}$ and $\approx^{\mathcal H}_{\text{we}}$ are useful in order to be able to obtain both an unwinding 
characterization in terms of $\approx^{\mathcal H}_{\text{we}}$ and one in terms of $\approx_{\text{we}}$.

\begin{lemma}\label{new2}
Let $P$ and  $Q$ be PEPA components that do not perform high level activities, then it holds $P\approx_{\text{we}} Q$ if and only if $P\approx_{\text{we}}^{\mathcal H}Q$.
\end{lemma}

As a consequence of Lemma \ref{new1} and Lemma \ref{new2} we get the following result.
\begin{corollary}\label{cornew1}
Let $P$ and  $Q$ be PEPA components. If $P\approx_{\text{we}}^{\mathcal H} Q$, then $P\setminus \mathcal H \approx_{\text{we}} Q\setminus \mathcal H$.
\end{corollary}}

Finally, we provide a characterization of $\mathit{EPSNI}$ in terms of \emph{unwinding conditions} \cite{VMCAI03,entcs-mefisto}. In practice, whenever a state $P'$ of a $\mathit{EPSNI}$ PEPA model $P$ may execute a high level activity leading it to a state $P''$, then $P'$ and $P''$ are indistinguishable in the sense that both $P'\approx_{\text{we}}^\mathcal{H} P''$ and $P'\setminus \mathcal H \approx_{\text{we}} P''\setminus \mathcal H$.



\section{Properties of Persistent Stochastic Non-Interference}\label{sec:compositionality}
\riccardo{
In this section, we prove some interesting properties of $\mathit{EPSNI}$. 
\begin{proposition} Let $P$ and $Q$ be two PEPA components. If $P,Q\in \mathit{EPSNI}$, then
\begin{itemize}
\item $P\setminus A\in \mathit{EPSNI}\,$ for all $A\subseteq {\cal A}$
\item $P/A\in \mathit{EPSNI}\,$ for all $A\subseteq {\cal A}$
\end{itemize}
\end{proposition}
\begin{proof}

Proofs are exactly the same as the one provided in \cite{PSNI_strong}. The major difference is the absence of the prefix operator for which the notion of weak-exact equivalence we introduce in this work is not compositional.
\end{proof}}

\section{Conclusion}\label{sec:conclusion}

In \cite{PSNI_strong}, the authors introduced and investigated the notion of Persistent Stochastic Non-Interference (PSNI), a security property aimed at preventing any leak of information from high-level to low-level components. At its core, PSNI relies on an equivalence relation on PEPA components called lumpable bisimulation.

In this work, we revisited PSNI and proposed a new variant, Exact Persistent Stochastic Non-Interference (EPSNI), grounded on the weak-exact equivalence relation. EPSNI refines the treatment of internal actions while preserving the quantitative accuracy of observable behaviours, thus providing a more precise semantic foundation for reasoning about non-interference in stochastic systems.

We presented two characterisations of EPSNI: a semantic formulation based on high-level contexts and a bisimulation-style definition that avoids explicit universal quantification. These perspectives not only clarify the meaning of EPSNI but also provide complementary proof techniques for establishing the property.

We further showed that EPSNI enjoys strong compositionality properties. In particular, the property is preserved under hiding and restriction on low-level actions, thereby supporting modular reasoning in the construction of complex systems.

To the best of our knowledge, this work constitutes the first application of weak-exact equivalence in a security setting, highlighting its robustness as a foundation for quantitative non-interference.

As future work, we plan to investigate the behaviour of EPSNI with the cooperation combinator, and to develop an unwinding-based verification strategy for EPSNI.

\vspace{0.5cm}
\noindent
\textbf{Acknowledgments.}
The work described in this
paper has been partially supported by PRIN project NiRvAna---Noninterference
and Reversibility Analysis in Private Blockchains” (20202FCJMH, CUP G23C22000400005) and by PNRR project SERICS---Security and Rights in the CyberSpace (CUP H73C22000890001.

\newpage
\small
\bibliographystyle{splncs04}
\bibliography{performance,security}

\newpage
\appendix

\end{document}